%% file: Verification-of-UDPP-is-undecidable.tex
\documentclass[a4paper,UKenglish,cleveref, autoref, thm-restate]{lipics-v2021}
\pdfoutput=1 %uncomment to ensure pdflatex processing (mandatatory e.g. to submit to arXiv)
\hideLIPIcs  %uncomment to remove references to LIPIcs series (logo, DOI, ...), e.g. when preparing a pre-final version to be uploaded to arXiv or another public repository

\usepackage{thm-restate}
\input{chairmacros.tex}

\input{notation.tex}

\DeclareMathOperator{\D}{\mathbb{D}}

\newcommand{\tikzcircle}[1]{\tikz\draw[#1,fill=#1] (0,0) circle (.5ex);}
\DeclareMathOperator{\ActiveA}{Act}
\DeclareMathOperator{\increment}{inc}
\DeclareMathOperator{\decrement}{dec}
\DeclareMathOperator{\goto}{goto}
\DeclareMathOperator{\nextA}{next}

\title{Verification of Population Protocols with Unordered Data is undecidable} %TODO Please add

\author{Roland Guttenberg}{Technical University of Munich, Germany}{guttenbe@in.tum.de}{0000-0001-6140-6707}{}%TODO mandatory, please use full name; only 1 author per \author macro; first two parameters are mandatory, other parameters can be empty. Please provide at least the name of the affiliation and the country. The full address is optional. Use additional curly braces to indicate the correct name splitting when the last name consists of multiple name parts.

\authorrunning{Roland Guttenberg} %TODO mandatory. First: Use abbreviated first/middle names. Second (only in severe cases): Use first author plus 'et al.'

\Copyright{Roland Guttenberg} %TODO mandatory, please use full first names. LIPIcs license is "CC-BY";  http://creativecommons.org/licenses/by/3.0/

 %TODO mandatory: Please choose ACM 2012 classifications from https://dl.acm.org/ccs/ccs_flat.cfm 

 %TODO mandatory; please add comma-separated list of keywords

\category{} %optional, e.g. invited paper

\relatedversion{} %optional, e.g. full version hosted on arXiv, HAL, or other respository/website
\funding{The project has received funding from the European Research Council (ERC) under the European Union’s Horizon 2020 research and innovation programme under grant agreement No 787367.}%optional, to capture a funding statement, which applies to all authors. Please enter author specific funding statements as fifth argument of the \author macro.

\acknowledgements{}%optional

\nolinenumbers %uncomment to disable line numbering

\EventEditors{John Q. Open and Joan R. Access}
\EventNoEds{2}
\EventLongTitle{42nd Conference on Very Important Topics (CVIT 2016)}
\EventShortTitle{CVIT 2016}
\EventAcronym{CVIT}
\EventYear{2016}
\EventDate{December 24--27, 2016}
\EventLocation{Little Whinging, United Kingdom}
\EventLogo{}
\SeriesVolume{42}
\ArticleNo{23}
\begin{document}

\maketitle

%TODO mandatory: add short abstract of the document
\begin{abstract}
We consider the new extension of population protocols with unordered data and show that the corresponding well-specification problem and therefore also other verification problems are undecidable.
\end{abstract}

\section{Introduction}
\label{Introduction}
Population protocols are a model of computation in which indistinguishable, mobile finite-state agents, randomly interact in pairs to decide whether their initial configuration satisfies a given property, modelled as a predicate on the set of all configurations~\cite{AADFP06}. The decision is taken by \emph{stable consensus}; eventually all agents agree on whether the property holds or not, and never change their mind again. 
Population protocols are very close to chemical reaction networks, a model in which agents are molecules and interactions are chemical reactions. In particular, Population Protocols have surpassed their original motivation of mobile sensor networks and are a well-studied model \cite{DotyE19, AAE08, AAER07, AADFP06, AADFP04, BlondinEGHJ20, CzernerGHE22, BlondinEJ18, AlistarhGV15, AlistarhG18, ElsasserR18, DotyEGSSU21, KosowskiU18}. However, the population protocol model has several variations: Existence of an initial leader to coordinate the computation \cite{AAE08}; a different notion of time when the protocol has finished \cite{KosowskiU18} which allows for faster protocols; as well as more than constant space \cite{AlistarhGV15, DotyEGSSU21}.

One of the main advantages of the main model of population protocols is that verification properties are decidable \cite{EsparzaGLM17}, even in the presence of leaders. On the other hand, with enough super-constant space, verification issues become undecidable. Hence this year Blondin \cite{Blondin23} considered another extension, which attempts to achieve a middle ground, and is the subject of this paper: By allowing for a read-only input of arbitrary size while preserving the fact that agents are finite-state, one might conjecture that verification properties stay decidable. However, in this paper we eliminate this hope by proving that well-specification is undecidable for the new model of population protocols with unordered data.

\section{Preliminaries}

In standard population protocols, finite-state agents interact in pairs in order to decide a property of their initial configuration by stable consensus subject to some fairness constraint. In the extension we consider in this paper, population protocols with unordered data, every agent in addition starts with read-only input data from some infinite domain \(\D\). Otherwise the semantics are the same.

\begin{definition}
A population protocol with unordered data (UDPP) over some infinite domain \(\D\) equipped with equality is a 4-tuple \(\Prot=(Q,\delta,I,O)\), where 
\begin{itemize}
\item \(Q\) is a finite set of states,
\item \(\delta \subseteq \N^{Q^2} \times \{=,\neq\} \times \N^{Q^2}\) is a finite set of transitions,
\item \(I \subseteq Q\) is a set of initial states.
\item \(O \colon Q \to \{0,1\}\) is an output mapping, assigning opinions to states.
\end{itemize}
\end{definition}

We call elements of the data domain \(\D\) colors. A configuration of a UDPP is a function \(Q \times \D \to \N\) with finite support, assigning to every state \(q\in Q\) and color \(d \in \D\) the number of agents which carry this combination. Finite support guarantees that configurations consist of finitely many agents. We call a state \(q\in Q\) active at a configuration \(C\) if there exists a color \(d \in \D\) such that \(C(q,d)\geq 1\). The set of active states is denoted \(\ActiveA(C)\). A configuration is initial if \(\ActiveA(C) \subseteq I\), i.e. all agents are in initial states, but potentially have different colors.

On configuration we define an addition via \((C+C')(q,d)=C(q,d)+C'(q,d)\), i.e. component-wise addition. Given a state \(q\in Q\) and a color \(d\in \D\) we denote by \(q_d\) the configuration \((q,d) \mapsto 1, (q',d') \mapsto 0\) otherwise. A transition \(t=((p,p'), \sim, (q,q')) \in \delta\) is enabled at a configuration \(C\) if there exist colors \(d \sim e\) such that \(C \geq p_d + p'_e\). The transition then leads to the configuration \(C'=C-(p_d+p'_e)+(q_d+q'_e)\), i.e. agents in states \(p,p'\) change to states \(q,q'\) without changing their color. By \(C \to C'\) we denote that \(C\) leads to \(C'\) via some transition \(t\in \delta\). We denote the reflexive and transitive closure of \(\to\) by \(\to^{\ast}\).

An execution \(\pi\) of a UDPP is a sequence \(\pi=(C_0,C_1, \dots)\) such that \(C_0\) is an initial configuration and \(C_i \to C_{i+1}\) for all \(i\). UDPPs accept by stable consensus, i.e. \(\pi\) has output \(b \in \{0,1\}\) if there exists \(n \in \N\) such that for all \(C'\) reachable from \(C_n\), we have \(\ActiveA(C') \subseteq O^{-1}(b)\), i.e. all states occurring in configuration \(C'\) have output \(b\). Observe that executions might have neither output \(1\) nor output \(0\), namely if, informally speaking, some agents forever disagree on the answer. Note that formally, we require a property about every reachable configuration; i.e. no matter what happens in the future, agents continue to believe in their current opinion. The execution \(\pi\) is fair if for every configuration \(C\), if \(C_i \to^{\ast} C\) for infinitely many \(i\), then \(C_i=C\) for infinitely many \(i\). I.e. if a configuration is infinitely often reachable, it is eventually reached. 

An initial configuration \(C_0\) has output \(b\) if every fair execution starting at \(C_0\) has output \(b\). If every initial configuration \(C_0\) has some output \(b=\varphi(C_0) \in \{0,1\}\), then the UDPP \(\Prot\) is called well-specified, and computes the predicate \(\varphi\).

\begin{example}
Consider \(\D:=\{\tikzcircle{red}, \tikzcircle{blue}, \tikzcircle{green}, \dots\}\). Define \(I=Q:=\{p,q\}\) and consider the transitions \(t_1=((p,q), \neq, (q,q))\) and \(t_2=((q,q),=, (p,q))\). Define \(O(p)=0\) and \(O(q)=1\). At the initial configuration \(C_0:=\{(p,\tikzcircle{red}) \mapsto 2, (q, \tikzcircle{blue}) \mapsto 1\}\) only \(t_1\) is enabled, and leads to \(C_1:=\{(p,\tikzcircle{red}) \mapsto 1, (q, \tikzcircle{red}) \mapsto 1, (q, \tikzcircle{blue}) \mapsto 1\}\). Now \(t_2\) is still disabled, because the two agents in state \(q\) have different colors. Hence \(t_1\) occurs again and leads to \(C_2:=\{(q,\tikzcircle{red}) \mapsto 2, (q, \tikzcircle{blue}) \mapsto 1\}\). At \(C_2\) only \(t_2\) is enabled and reverses the effect of \(t_1\), hence there is only one execution, cycling between \(C_1\) and \(C_2\). The execution does not have an output, as there is always an agent in \(q\) with output \(1\), but also infinitely often an agent in \(p\) with output \(0\). Since the execution does not have an output, \(C_0\) does not have an output and \(\Prot\) is not well-specified.
\end{example}

In this paper we prove that the following decision problem is undecidable:

\begin{restatable}{proposition}{PropositionVerificationUndecidable}
The following problem is undecidable:

Input: Population Protocol with unordered data \(\Prot\).

Output: Is \(\Prot\) well-specified?
\end{restatable}

\section{Proof of our main result}

We will reduce from the special halting problem for counter machines.

\begin{definition}
A 2-counter machine, where we call the counters \(x\) and \(y\), is a finite list of instructions, where every instruction is one of the following:
\begin{itemize}
\item \(\increment(c)\), \(c \in \{x,y\}\): Increment counter \(c\),
\item \(\decrement(c,k)\), \(c \in \{x,y\}\), \(k \in \N\): If \(c>0\), then decrement \(c\), otherwise goto instruction \(k\),
\item \(\goto(k)\), \(k \in \N\): Goto instruction \(k\),
\item Halt.
\end{itemize}
\end{definition}

Let \(n\) be the number of instructions, define \(Q_{CM}:=\{1,\dots, n\}\). The semantics are as follows: A configuration of a 2-counter machine is a triple \((q, x,y) \in Q_{CM} \times \N \times \N\) of an instruction number, and current values of counters \(x\) and \(y\). The initial configuration is \((1,0,0)\), i.e. we start at the first instruction with empty counters. A transition performs the obvious change to the counters and increases the current instruction number by 1.

It is well known \cite{Minsky67} that whether a 2-counter machine eventually reaches the halt instruction is undecidable. We will hence proceed by reduction from the halting problem.

\begin{restatable}{proposition}{ReductionFromCounterMachineHalting}
The halting problem for 2-counter machines polynomially reduces to the complement of the well-specification problem for population protocols with unordered data.
\end{restatable}

\begin{proof}
We start with an outline of the basic ideas, before starting the formal proof.

\textbf{Intuition}: Except for the possibility to cheat and apply an if-zero transition at non-zero counter values, population protocols can simulate counter machines. Namely one uses states \(Q:=Q_{CM} \cup \{x,y\}\) and starts with one agent in the first instruction. Counter values of \(x\) and \(y\) are represented by numbers of agents in the corresponding states. The transitions (except for guaranteeing \(=0\)) correspond exactly to population protocol transitions taking one agent representing the active instruction, possibly changing counters and moving to the next instruction. Hence many impossibility results for population protocols or also Petri nets were obtained by building an anti-cheat mechanism, for example by Jancar in \cite{Jancar95}. We have to build a similar mechanism in our reduction. 

One main trick will be to use shadow states \(\{\bar{x}, \bar{y}\}\), which always contain exactly one agent, which stores the color that all tokens in \(x\) respectively \(y\) are supposed to have in the current configuration. Similar to Rosa-Velardo in \cite{Rosa-VelardoF11}, whenever we check for zero we change this shadow color into a ``fresh'' color, which has never occurred before. In case that we cheated with the zero test, some agents with the now wrong color will be left forever. Hence at every future configuration, these two agents have the chance to meet, and by fairness eventually do; aborting the computation. Preserving the fact that a cheat happened and eventually observing the violation by fairness is used two more times in our argument.

Our target UDPP will have two initial states \(R_1, R_2\). \(R_1\) will be used as arbitrarily large reservoirs for incrementing counters. However, \(R_2\) is the crucial initial state ensuring the above freshness: In \(R_2\), we assume that every color \(d\in \D\) initially occurs at most once, otherwise we will again have an observable violation. We ensure that this violation can be observed forever by having tokens remember their initial states. Since colors also do not change, two violating tokens will forever be able to determine that they form a violation.

The condition on \(R_2\) allows us to guarantee introducing a fresh color into the computation. Namely, we will keep the invariant that the computing agents do not have any common colors with agents from \(R_2\). This is of course only possible because introducing an agent from \(R_2\) into the computation has deleted this color from \(R_2\), since it occurs only once.

A third violation which can occur in our simulation is that we might have two counter machine instructions active at the same time. This is because all initial states might be filled with arbitrarily many agents, hence there is no way to guarantee that exactly one agent will control the machine. However, similar to above, we will guarantee that agents in states \(Q_{CM}\) will stay in \(Q_{CM}\) forever; hence eventually they will meet and recognize the violation.

Finally, well-specification will be connected to halting as follows: Input states are the only states with output \(1\), and never refilled. The first transition will put an agent into a state \(sink_1\) which will empty the initial states, and hence cause well-specification. This process can only be stopped by an agent in a halt instruction, which empties \(sink_1\). This is similar to a reduction from Petri net reachability to verification of standard population protocols by Esparza et. al. \cite{EsparzaGLM17}.

\textbf{Formal Protocol}: Let \(CM\) be a counter machine with \(n\) instructions \(i_1, \dots, i_n\). For every instruction \(i_m=\decrement(c,k)\) for \(c\in \{x,y\}, k \in \N\), we also add an intermediate instruction \(i_m'\). Let \(Q_{CM}:=\{i_1,\dots, i_n\} \cup \{i_m' \mid i_m=\decrement(c,k)\}\). We define a population protocol with unordered data \(\Prot\) over some infinite domain \(\D\) as follows:
\begin{enumerate}
\item \(Q_{main}:=Q_{CM} \cup \{x, y\} \cup \{\bar{x}, \bar{y}\} \times \{+,-,=0,>0\} \cup \{setup\} \times \{x,y\} \cup \{R_1,R_2, sink_1, sink_2, garbage\}\).
\item \(Q:=Q_{main} \times \{R_1,R_2\}\).
\item \(\delta\) will be defined step by step later.
\item \(I=\{(R_1,R_1),(R_2,R_2)\}\).
\item \(O(\{R_1,R_2\} \times \{R_1, R_2\})=1\), else \(O(p,R_i)=0\) for all \(p \in Q_{main} \setminus \{R_1,R_2\}, i \in \{1,2\}\).
\end{enumerate}

The actual states of the protocol are \(Q\), but the second component is never updated; it simply remembers the initial state. We will hence mainly refer to \(Q_{main}\) and specify all except one transition in terms of \(Q_{main}\) only. We start with this exceptional transition, for \(p,p',q,q' \in Q\) we denote \(((p,p'), \sim, (q,q')) \in \delta\) by \(p,p' \mapsto_{\sim} q, q'\), or leave away \(\sim\) if this transition is enabled in either case:

\begin{align}
\tag{Input Violation} \label{tra:Violation1}
&(p, R_2), (p', R_2), \mapsto_{=} (sink_2, R_2), (sink_2, R_2) & p, p' \in Q_{main}
\end{align}

This transition realizes that a violation occurred; multiple agents of the same color started in \(R_2\). It overwrites whatever other transition would be defined between the states in \(Q_{main}\). We continue by specifying the other violation transitions; namely if two agents are in \(\bar{x}\), or if an agent in \(\bar{x}\) meets an agent in \(x\) of different color:

\begin{align}
\tag{Counter Color Violation} \label{tra:violation2}
&(\bar{c}, b), c \mapsto_{\neq} sink_2, sink_2 & b \in \{+,-,=0,>0\},  \in \{x,y\} \\
\tag{Control State Violation} \label{tra:violation3}
&(\bar{x}, b), (\bar{x}, b') \mapsto sink_2, sink_2 & b, b' \in \{+,-,=0,>0\} \\
\tag{Convert To Sink1} \label{tra:convert1}
&sink_1, R_i \mapsto sink_1, sink_1 & i \in \{1,2\} \\
\tag{Convert To Sink2} \label{tra:convert2}
&sink_2, q \mapsto sink_2, sink_2 & q \in Q_{main}
\end{align}

The \eqref{tra:convert2} transition informs the other agents about violations, the \eqref{tra:convert1} transition guarantees that as long as \(sink_1\) is not emptied, eventually all initial states are emptied and it is always possible to reach a stable consensus for output \(0\). One could also implement a transition like \eqref{tra:violation3} for \(\bar{y}\), or for two agents in \(Q_{CM}\), but this is not necessary. Whenever there is a \(\bar{y}\) violation, then there is also a \(\bar{x}\) violation.

Next we explain the actual counter machine simulation. Initially, we set up one agent for controlling the counter machine, as well as one agent in \((\bar{x},=0)\) and \((\bar{y},=0)\) respectively. These agents are then responsible for the actual increments and decrements by moving to \(\{+,-\}\) respectively, and enable decrementing by observing a non-empty counter. Furthermore, they guarantee that all agents in \(x\) in any faithful simulation have the same color as themselves. The setup is done via the following transitions:

\begin{align}
\tag{Setup1} \label{tra:setup1}
& R_1, R_2 \mapsto (setup,x), sink_1\\
\tag{Setup2} \label{tra:setup2}
& (setup,x), R_2 \mapsto (setup,y), (\bar{x}, =0) \\
\tag{Setup3} \label{tra:setup3}
& (setup,y), R_2 \mapsto i_1, (\bar{y}, =0)
\end{align}

The increments and decrements are performed as follows, where for conciseness we specify them for \(c \in \{x,y\}\), i.e. for both counters at once. The third instruction simply detects that the counter is non-zero, enabling later decrements:

\begin{align}
\tag{Increment c} \label{tra:increC}
& (\bar{c}, +), R_1 \mapsto_{=} (\bar{c}, >0), c & c \in \{x,y\}\\
\tag{Decrement c} \label{tra:decreC}
& (\bar{c}, -), c \mapsto_{=} (\bar{c}, =0), garbage & c\in \{x,y\} \\
\tag{Detect >0} \label{tra:DetectNonZeroC}
& (\bar{c}, =0), c \mapsto_{=} (\bar{c}, >0), c & c \in \{x,y\}
\end{align}

At last we specify the transitions performing counter machine instructions. To not require an extra case for \(\goto(k)\) instructions, we define \(\nextA(m):=k\) if the next instruction is \(\goto(k)\), and \(\nextA(m):=m+1\) otherwise, i.e. we shortcut the goto.

Increment simply informs the responsible agent that he is supposed to increment.

\begin{align}
\tag{Inc(c)}
& i_m, (\bar{c}, b) \mapsto i_{\nextA(m)}, (\bar{c},+) & i_m=\increment(c), c\in \{x,y\}, b \in \{=0, >0\}
\end{align}

To simulate the decrement instructions, we need three transitions each: Two for the \(=0\) case and one for the \(> 0\) case. In the \(>0\) case, correct simulation requires the responsible agent to determine via \eqref{tra:DetectNonZeroC} that the counter is actually \(>0\); afterwards it is simply told to decrement. In the \(=0\) case, the old shadow color is removed and replaced with a color from the stack of fresh colors \(R_2\).

\begin{align}
\tag{Dec(c)} \label{tra:decrementC}
& i_m, (\bar{c}, >0) \mapsto i_{\nextA(m)}, (\bar{c}, -) & i_m=\decrement(c,k), c\in \{x,y\}, k\in \N \\
\tag{c=0, Part 1}
& i_m, (\bar{c}, =0) \mapsto i_m', garbage & i_m=\decrement(c,k), c\in \{x,y\}, k \in \N \\
\tag{c=0, Part 2}
& i_m', R_2 \mapsto i_k, (\bar{c}, =0) & i_m=\decrement(c,k), c\in \{x,y\}, k \in \N
\end{align}

We have finished specifying the protocol after one more transition, causing the halt. The idea is that \eqref{tra:setup3} has initially filled \(sink_1\); this causes well-specification by removing input agents. After reaching the halt instruction we empty \(sink_1\) to reach a deadlock without matching opinions.

\begin{align}
\tag{Cause Deadlock} \label{tra:DeadlockIfHalt}
& i_m, sink_1 \mapsto i_m, garbage & i_m=halt
\end{align}

\textbf{Correctness}: First assume that the counter machine \(CM\) does halt.

Let \(k\) be the number of steps \(CM\) requires to halt. Let \(C_0\) be an initial configuration with \(C_0(d_i,R_1)\geq 2k+7\) agents in \(2k\) different colors \(d_1, \dots, d_{2k}\). Moreover, assume \(C_0(d_i, R_2)=1\) for at least \(2k\) different colors \(d_1', \dots, d_{2k}'\) and assume that \(C_0(d, R_2) \leq 1\) for all \(d\in \D\), i.e. there is no input violation.

We perform the following transition sequence \(\sigma\): We start with \eqref{tra:setup1}, \eqref{tra:setup2}, \eqref{tra:setup3}, followed by the correct counter machine simulation until it halts. Observe that any simulating transition takes at most \(2\) agents out of \(R_1\) and at most one new color from \(R_2\), hence we do not run out of agents in \(R_1, R_2\). Call the reached configuration \(C\). We have \(C(d, R_1) \geq 1\) for some color \(d\), since we used at most \(2k+6\) agents so far. Apply transition \eqref{tra:convert1} until no agents are in state \(R_2\), such that transition \eqref{tra:setup1} becomes disabled. Now apply transition \eqref{tra:DeadlockIfHalt} until the state \(sink_1\) is empty. The reached configuration is a deadlock, and there is both an agent in \(R_1\) and an agent in the halt instruction. Hence the configuration does not have an output. Hence \(\Prot\) is not well-specified.

Now assume that \(CM\) does \emph{not halt}. We start with four simple but important observations: 

\begin{enumerate}
\item Agents who enter \(\{\bar{x}, \bar{y}\} \times \{+,-,=0,>0\}\) are always from \(R_2\), a previously unused color.
\item Agents who enter or leave \(x,y\) are always the same color as the corresponding \(\bar{x}, \bar{y}\).
\item State \(sink_1\) can only be emptied using \eqref{tra:convert2} or \eqref{tra:DeadlockIfHalt}.
\item The only states with output \(1\) are the input states \(R_1\) and \(R_2\) which have no incoming transitions, hence any \(0\)-consensus is automatically stable.
\end{enumerate}

We have to prove that every initial configuration has a unique output. Let \(C_0\) be any initial configuration. Assume that \eqref{tra:Violation1} is disabled at \(C_0\). Otherwise it will always be enabled, hence eventually occur and lead to a configuration with all agents in \(sink_2\), having output 0. 

Under this assumption, the only transition which can be enabled at \(C_0\) is \eqref{tra:setup1}. If it is not enabled, then \(C_0\) is a deadlock with output 1, since every initial state has output 1. Hence assume that \eqref{tra:setup1} is enabled. 

In this case, we claim that every execution \(\pi=(C_0,C_1, \dots)\) stabilizes to output 0. Assume that in \(\pi\) no violation occurs, i.e. \(C_m(sink_2)=0\) for all \(m \in \N\), otherwise \(\pi\) has output \(0\) via transition \eqref{tra:convert2}. Since no other transition is enabled, first transition \eqref{tra:setup1} occurs, causing \(C_1(sink_1)>0\). Since \(sink_1\) removes all agents of opinion \(1\) via \eqref{tra:convert1}, if we do not empty it, then by observation 4 we stabilize to opinion \(0\) as required. By observation 3, \(\pi\) can only empty \(sink_1\) by eventually placing a token in a halting instruction. Hence it suffices to show that no agent ever moves to a halting instruction, i.e. \(C_m(i_k)=0\) for all \(m\in \N\) and halt instructions \(i_k\). 

Before any other transitions are enabled, \eqref{tra:setup2} and \eqref{tra:setup3} have to occur. In fact they occur exactly once, as otherwise there would be agents in \((\bar{x},b)\) and \((\bar{x},b')\) for some \(b,b'\) and \eqref{tra:violation3} is infinitely often enabled, and would eventually occur, contradiction to no violations.

This implies that there is always exactly one active counter machine instruction. Assume that eventually, some instruction is executed wrong; the only possibility for this is that at some \(C_m\) the goto \(k\) part of a \(\decrement(c,k)\) instruction was applied with \(C_m(c,d)>0\) for some \(d\in \D\). However, by observations 1 and 2, we then have \(C_l(c,d)>0\) for all \(l \geq m\), hence \eqref{tra:violation2} is infinitely often enabled, and would eventually occur, contradiction.

Hence the counter machine simulation has to be faithful to not cause violations. Since \(CM\) does not halt, we have \(C_m(i_k)=0\) for all \(m \in \N\) and halt instructions \(i_k\) as required. Remember that this was sufficient as \(sink_1\) then removes all agents of opinion 1.
\end{proof}

\bibliography{Verification-of-UDPP-is-undecidable.bib}

\end{document}

%% file: chairmacros.tex
\usepackage{xcolor}
\usepackage{stmaryrd}
\usepackage{tikz}
\usetikzlibrary{automata, positioning, arrows, petri, backgrounds}

\newcommand{\N}{\mathbb{N}}

\definecolor{niceredbright}{HTML}{bd0310}
\definecolor{nicebluebright}{HTML}{197b9b}
\definecolor{nicered}{HTML}{7f0a13}
\definecolor{niceblue}{HTML}{104354}
\definecolor{nicegreen}{HTML}{217516}
\definecolor{nicepurple}{HTML}{884bab}
\definecolor{nicebg}{HTML}{f6f0e4}
\definecolor{niceredlight}{HTML}{c9888d}
\definecolor{nicebluelight}{HTML}{78a4b8}
\definecolor{nicegreenlight}{HTML}{76de68}
\definecolor{nicepurplelight}{HTML}{bc87db}

%% file: notation.tex
\newcommand{\Prot}{\mathcal{P}} % Protocol